\documentclass{amsart}

\usepackage{tikz}

\newtheorem{theorem}{Theorem}[section]

\newcommand{\tH}{\tilde{H}}

\newcommand{\Z}{\mathbb{Z}}
\newcommand{\MD}{\mathbb{MD}}
\newcommand{\REXT}{{\mathbb{REXT}}}
\newcommand{\cM}{\mathcal{M}}
\newcommand{\cZ}{\mathcal{Z}}

\newcommand{\mK}{K}
\newcommand{\hcZ}{\widehat{\cZ}}

\newcommand{\tM}{\tilde{M}}
\newcommand{\hM}{\hat{M}}

\newcommand{\myboxdot}{\rlap{\hskip0.25ex\raise0.2ex\hbox{$\bullet$}}\square}

\newcommand{\emptybox}{\hbox{$\square$}}

\newcommand{\Wr}{\operatorname{Wr}}

\begin{document}
\title{Ladder operators and rational extensions}
\author{David G\'omez-Ullate}
\address{Escuela Superior de Ingenier\'ia, U.  C\'adiz, 11519 Puerto Real,
  Spain, \and
  Departamento de F\'isica Te\'orica, U. Complutense, 28040 Madrid, Spain}
\email{david.gomezullate@uca.es}

\author{ Yves Grandati}
\address{ L. Physique et Chimie Th\'eoriques, U. de
 Lorraine, 57078 Metz, Cedex 3, France}
\email{yves.grandati@univ-lorraine.fr}

\author{Zo\'e McIntyre}
\address{Department of Physics, McGill U., Montr\'eal QC
  Canada H3A 2T8}
\email{zoe.mcintyre@mail.mcgill.ca}

\author{Robert Milson}
\address{Dept. of Mathematics and Statistics, Dalhousie U.,
  Halifax NS, Canada B3H 3J5} \email{rmilson@dal.ca}

%\authorrunning{Gomez-Ullate, Grandati, McIntyre, Milson}

\begin{abstract}
  This note presents the classification of ladder operators
  corresponding to the class of rational extensions of the harmonic
  oscillator.  We show that it is natural to endow the class of
  rational extensions and the corresponding intertwining operators
  with the structure of a category $\REXT$.  The combinatorial data
  for this interpretation is realized as a functor $\MD \to \REXT$,
  where $\MD$ refers to the set of Maya diagrams appropriately endowed
  with categorical structure.  Our formalism allows us to easily
  reproduce and extend earlier results on ladder operators.
\end{abstract}

\maketitle

%\keywords{Rational extensions, Ladder operators, Maya diagrams}

\section{Introduction}

Supersymmetric quantum mechanics (SUSYQM) has proven to be a key
technique in the construction of exactly-solvable potentials and in
the understanding of shape-invariance.  The supersymmetric partners of
the harmonic oscillators are known as rational extensions because the
corresponding potentials have the form of a harmonic oscillator plus a
rational term that vanishes at infinity.

There has been some recent interest in rational extensions possessing
ladder operators, which may be thought of as higher order analogues of
the classical creation and annihilation operators.  There are
applications of such ladder operators to superintegrable systems
\cite{MQ1,MQ2}, rational solutions of Painlev\'e equations \cite{MN},
and coherent states \cite{HHMZ}.

In this note we classify the ladder operators corresponding to the
class of rational extensions of the harmonic oscillator.  Rational
extensions are naturally associated with combinatorial objects called
Maya diagrams.  We show that any two rational extensions are related
by an intertwining relation.  It therefore makes sense to endow both
Maya diagrams and rational extensions with the structure of a
category, and to interpret the relation Maya diagram $\mapsto$
rational extension as a functor between these categories.  This
approach allows us to classify ladder operators and syzygies of ladder
operators, and thereby to generalize the results of \cite{MQ1,MQ2}.

\section{Maya diagrams}

A Maya diagram is a set of integers $M\subset \Z$ containing a finite
number of positive integers, and excluding a finite number of negative
integers.  We visualize a Maya diagram as a horizontally extended
sequence of $\myboxdot$ and $\emptybox$ symbols, with the filled symbol
$\myboxdot$ in position $m$ indicating membership $m\in M$. The
defining assumption now manifests as the condition that a Maya diagram
begins with an infinite filled $\myboxdot$ segment and terminates with
an infinite empty $\emptybox$ segment.

A Maya diagram may also be regarded as a strictly decreasing sequence
of integers $m_1 > m_2> \cdots$, subject to the constraint that
$m_{i+1} = m_i-1$ for $i$ sufficiently large.  It follows that there
exists a unique integer $\sigma$, called the index of $M$, such that
$m_i = -i+\sigma$ for $i$ sufficiently large.

Let $\cM$ denote the set of all Maya diagrams.  The flip at position
$k\in \Z$ is the involution $f_k:\cM\to \cM$ defined by
\begin{equation}\label{eq:flipdef}
 f_k : M \mapsto
\begin{cases}
   M \cup \{ k \}, & \text{if}\quad k\notin M, \\
   M \setminus \{ k \},\quad & \text{if}\quad k\in M.
\end{cases}\qquad M\in \cM.
\end{equation}
\noindent
In the first case, we say that the flip acts on $M$ by a
state-deleting transformation ($\emptybox\to$ $\myboxdot$), and in the
second case, by a state-adding transformation
($\myboxdot$$\to\emptybox$).

Let $\cZ_p$ denote the set of subsets of $\Z$ having cardinality $p$,
and $\cZ = \bigcup_p \cZ_p$ the set of all finite subsets of $\Z$.
For $K\in \cZ_p$ consisting of distinct $k_1,\ldots, k_p\in \Z$ we
define the multi-flip $f_K:\cM \to \cM$ by
\begin{equation}
  \label{eq:fKMdef}
 f_K(M) = (f_{k_1} \circ \cdots \circ f_{k_p})(M),\quad M\in \cM,
\end{equation}
Since flips commute, the action of $f_K$ does not depend upon the
order of $k_1,\ldots, k_p$.

It is useful to regard $\cM$ as a complete graph whose edges are
multi-flips.  For Maya diagrams $M_1,M_2\in \cM$, the symmetric
difference
\[ M_1\ominus M_2 = (M_1\setminus M_2) \cup (M_2 \setminus M_1) \]
is precisely the edge that connects $M_1$ and $M_2$.  More precisely, if
\[ K = M_1\ominus M_2 = M_2 \ominus M_1, \] then $f_K(M_1) = M_2$ and
$f_K(M_2) = M_1$.

Multi-flips can also be used to define a bijection $\cZ\to\cM$ given by
$ K\mapsto f_K(M_\emptyset)$, where $M_\emptyset:=\Z_{-}$
denotes the trivial Maya diagram.  We refer to $K\in \cZ$ as the index
set of the Maya diagram $f_K(M_\emptyset)$.

The additive group $\Z$ acts on $\cM$, because  for $M\in \cM$  and
$n\in \Z$, the set
\[ M+n = \{ m+n \colon m\in M \}\] is also a Maya
diagram. Moreover, we have
\begin{equation}
  \label{eq:indexshift}
  \sigma_{M+n}=\sigma_M+n.  
\end{equation}

We will refer to an equivalence class of Maya diagrams related by
translations as an \textit{unlabelled Maya diagram}, and denote the
set of all unlabelled Maya diagrams by $\cM/\Z$.  One can visualize
the passage from an unlabelled to a labelled Maya diagram as choosing
the placement of the origin.
% By
% \eqref{eq:indexshift}, every unlabelled Maya diagram permits a unique
% placement of the origin so as to obtain a Maya diagram with zero
% index.

For $B\in \cZ_p$, where $p=2g+1$ is odd, we define the Maya
diagram
\begin{equation}
  \label{eq:MBi}
  \Xi(B)= (-\infty,b_{0}) \cup [b_{1},b_{2}) \cup
  \ \cdots \cup [b_{2g-1},b_{2g}),
\end{equation}
where $b_0 <b_1<\cdots < b_{2g}$ is an increasing enumeration of $B$
and where $[m,n) = \{ j\in \Z \colon m\leq j < n\}$.  Every Maya
diagram has a unique representation of the form $\Xi(B)$ for some
$B\in \cZ_{2g+1}$.  We will call the corresponding $g\geq 0$ the genus
of $M= \Xi(B)$ and refer to $(b_0,\ldots, b_{2g})$ as the block
coordinates of $M$.  The block coordinates may also be characterized
as the unique set $B\in \cZ$ such that $f_B(M) = M+1$.

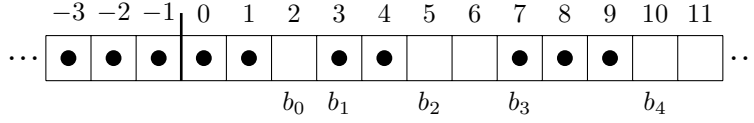
\begin{figure}[h]
\begin{tikzpicture}[scale=0.6]

\draw  (1,1) grid +(15 ,1);

\path [fill] (0.5,1.5) node {\huge ...}
++(1,0) circle (5pt) ++(1,0) circle (5pt)  ++(1,0) circle (5pt)
++(1,0) circle (5pt) ++(1,0) circle (5pt)
++(2,0) circle (5pt) ++(1,0) circle (5pt)
++ (3,0) circle (5pt)  ++(1,0) circle (5pt)   ++ (1,0) circle (5pt)
++ (3,0) node {\huge ...} +(1,0);
%node[anchor=west] { $M =  (-\infty,b_0)\cup [ b_1,b_2) \cup [ b_3,b_4)$};

\draw[line width=1pt] (4,1) -- ++ (0,1.5);

\foreach \x in {-3,...,11} 	\draw (\x+4.5,2.5)  node {$\x$};
\path (6.5,0.5) node {$b_0$} ++ (1,0) node {$b_1$}
++ (2,0) node {$b_2$}++ (2,0) node {$b_3$}++ (3,0) node {$b_4$}
;
\end{tikzpicture}
\caption{Block coordinates $(b_0,\ldots, b_4) = (2,3,5,7,10)$
  of a genus $2$ Maya diagram
  $M = (-\infty,b_0)\cup [ b_1,b_2) \cup [
  b_3,b_4)$. Note that the genus is both the number of
  finite-size empty blocks and the number of finite-size filled
  blocks.}\label{fig:genusM}
\end{figure}

Figure~\ref{fig:genusM} explains the visual meaning of block
coordinates and of genus.  After removal of the initial infinite
$\myboxdot$ segment and the trailing infinite $\emptybox$ segment, a
Maya diagram consists of alternating empty $\emptybox$ and filled
$\myboxdot$ segments of variable length.  The genus $g$ counts the
number of such pairs.  The even block coordinates $b_{2i}$ indicate
the starting positions of the empty segments, and the odd block
coordinates $b_{2i+1}$ indicate the starting positions of the filled
segments.

\section{Rational extensions}

For  $n\in \Z$, set
\[ \psi_n(x) =
  \begin{cases}
    e^{-\frac{x^2}{2}} H_n(x) & \text{ if } n\geq 0\\
    e^{\frac{x^2}{2}} \tH_{-n-1}(x) & \text{ if } n<0
  \end{cases}
\]
where,
\[
  H_n(x) = (-1)^n e^{x^2} \frac{d^n}{dx^n} e^{-x^2} ,\quad n=0,1,2,\ldots
\]
are the Hermite polynomials, and
\[ \tH_n(x) = (-\mathrm{i} )^{n} H_n(\mathrm{i} x) \] are the conjugate Hermite
polynomials. We then have
\[ -\psi_n''(x) + x^2 \psi_n(x) = (2n+1) \psi_n(x),\quad n\in \Z.\]
For $n\geq 0$, the above solutions correspond to the bound states of
the quantum harmonic oscillator.  The solutions for $n<0$ do not
satisfy the boundary conditions at $\pm\infty$ and therefore represent
virtual states.

For $M\in \cM$ with  index set $K\in \cZ_p$, 
let $s_1>\cdots > s_r\geq 0$ and $t_1>\cdots > t_q\geq 0$
be the uniquely specified lists of  natural numbers such  that 
\[ K = \{ -1-s_1,\ldots, -1-s_r, t_q,\ldots, t_1 \},\quad
  p=q+r.\]
We will refer to $(s_1,\ldots, s_r \mid t_q,\ldots ,t_1)$ as the
\textit{Frobenius symbol} of $M$. It is easy to check that the index
of $M$ is given by $\sigma = q-r$.

Let us now define
\begin{equation} H_M(x) = e^{\sigma_M \frac{x^2}{2}}\Wr[\psi_{k_1},\ldots, \psi_{k_p} ],
\end{equation}
where $k_1<\cdots < k_p$ is an increasing enumeration of $K$, where
$\sigma_M\in \Z$ is the index, and $\Wr$ is the usual Wronskian
determinant.  The polynomial nature of $H_M(x)$ becomes evident in the
following pseudo-Wronskian \cite{GGM2} realization:
\begin{equation}\label{eq:pWdef2} H_M =
  \begin{vmatrix} \tH_{s_1} & \tH_{s_1+1} & \ldots &
    \tH_{s_1+r+q-1}\\ \vdots & \vdots & \ddots & \vdots\\ \tH_{s_r} &
    \tH_{s_r+1} & \ldots & \tH_{s_r+r+q-1}\\ H_{t_q} & H'_{t_q} &
    \ldots & H^{(r+q-1)}_{t_q}\\ \vdots & \vdots & \ddots & \vdots\\
    H_{t_1} & H'_{t_1} & \ldots & H^{(r+q-1)}_{t_1}
  \end{vmatrix}.
\end{equation}

A suitably normalized pseudo-Wronskian is a translation invariant of
the underlying Maya diagram.  The following result was proved in
\cite{GGM2}.  Set
\begin{equation}
  \label{eq:hHdef}
  \widehat{H}_M =
  \frac{(-1)^{rq}H_M}{
    \prod_{i<j}  2(s_j-s_i)\prod_{i<j} 2( t_i-t_j)}.
\end{equation}
Then for $M\in \cM$ and $n\in\Z$ we have
\begin{equation} \label{eq:HMequiv}
  \widehat{H}_M =  \widehat{H}_{M+n}.
\end{equation}

The potential
\begin{equation}
  \label{eq:UMdef}
\begin{aligned}
  U_M(x)
  &= x^2 - 2 \frac{d^2}{dx^2}\log \Wr[ \psi_{k_1},\ldots, \psi_{k_p}
    ],\\ 
  &= x^2 + 2\left( \frac{H_M'}{H_M}\right)^2 - \frac{2H_M''}{H_M} - 2 \sigma_M
\end{aligned}
\end{equation}
is known as a rational extension \cite{GGM} of the harmonic
oscillator.  The corresponding Hamiltonian operators
\begin{equation}
  \label{eq:TMdef}
   T_M = -\frac{d^2}{dx^2} + U_M 
\end{equation}
 are exactly solvable with
\[ T_M[\psi_{M,k}] = (2k+1) \psi_{M,k},\]
where
\[
  \psi_{M,k}
  =  e^{ \frac{\epsilon x^2}{2}} \frac{H_{f_k(M)}}{H_M},\qquad
  \epsilon =
  \begin{cases}
    +1 & \text{ if } k\in M\\
    -1 & \text{ if } k\notin M
  \end{cases}.
\]
Note that, as a consequence of \eqref{eq:indexshift} and
\eqref{eq:HMequiv}, $T_M$ is translation covariant:
\begin{equation}
  \label{eq:TM+n}
  T_{M+n} = T_M  + 2n,\quad n\in \Z.
\end{equation}

Let $(b_0,b_1,\ldots, b_{2g})$ be the block coordinates of $M$.  By
the Krein-Adler theorem \cite{adler,GGM,krein}, the polynomial $H_M$
has no real zeros if and only if $b_{2j}-b_{2j-1}$ is even for all
$j=1,\ldots, g$, i.e., if all the finite $\myboxdot$ segments of $M$
have even size.  For such Maya diagrams, the potential $U_M$ is
non-singular and hence $T_M$ corresponds to a self-adjoint operator.
The bound states of the operator correspond to the empty boxes of $M$,
i.e., to $k\notin M$.  It is precisely for such $M\in \cM$ and
$k\notin M$ that the eigenfunction $\psi_{M,k}$ is
square-integrable. For such $M$ and $k$, the polynomial part of
$\psi_{M,k}$ is known as an exceptional Hermite polynomial \cite{GGM}.

\section{Categorical Structure}
In this section, we define $\MD$, a category whose objects are Maya
diagrams and whose arrows are multi-flips, and $\REXT$, another
category whose objects are rational extensions and whose arrows are
intertwining operators (definition given below).  We then
exhibit a functor $\MD \to \REXT$ that we use  to classify ladder
operators.

In order to define composition of arrows, it will first be necessary
to generalize the notion of a multi-flip.  A multi-set is a
generalized set object that allows for multiple instances of each of
its elements.
% Formally, a multi-set of integers, is a mapping
% $\mK:\Z\to \N_0$, where the latter denotes the set of non-negative
% integers.  The cardinality of a multi-set is defined to be
% $|\mK|=\sum_i \mK(i)$.  The union of two multi-sets %$\mK_1, \mK_2$
% is defined as their the point-wise sum.
% \[ (\mK_1 \cup \mK_2)(i) = \mK_1(i) + \mK_2(i),\quad i\in \Z.\]
Let $\hcZ_p$ denote the set of integer multi-sets of cardinality $p$
and $\hcZ = \bigcup_p \hcZ_p$ the set of finite integer multi-sets.
We express a multi-set $K\in \hcZ$ as
\begin{equation}
  \label{eq:Kki}
K = \{ k_1^{p_1},\ldots, k_q^{p_q}\} \,
\end{equation}
where $k_1,\ldots, k_q\in \Z$ are distinct, and where $p_i>0$ indicate
the multiplicity of element $k_i$.  The cardinality is then given by
$p=p_1+\cdots +p_q$.  The notion of a multi-flip extends naturally
from sets to multi-sets.  Indeed, for $\mK\in \hcZ$, we re-use
\eqref{eq:fKMdef} to define the multi-flip $f_{\mK}:\cM\to \cM$.

We say that $K$ is an even multi-set if all of its elements have an
even multiplicity.  Since flips are involutions, $f_K$ is the identity
transformation if and only if $K$ is even.  If $K$ is an even
multi-set then it has the unique decomposition $K=K_1\cup K_1$ where
$K_1$ has the same elements as $K$ but with the multiplicities divided
by $2$.  More generally, every multi-set $K\in \hcZ$ has a unique
decomposition of the form
\begin{equation}
  \label{eq:KK0K1}
  K = K_0 \cup K_1\cup K_1, \quad K_0\in \cZ,\; K_1 \in \hcZ,
\end{equation}
where $K_0$ is the set of integers that occur in $K$ with an odd
multiplicity.  Again, since flips are involutions, we have
$f_K = f_{K_0}$.

The objects of $\MD$ are labelled Maya diagrams $\cM$, and the arrows
are pairs $(M,K)\in \cM\times \hcZ$.
 The source of $(M,K)$ is $M$, and the target is $f_K(M)$.
Composition of morphisms is given by the union of multi-sets:
\[ (M_2,K_2) \circ (M_1,K_1) = (M_1, K_1 \cup K_2), \] where
$M_1\in \cM,\; K_1,K_2\in \hcZ,\; M_2 = f_{K_1}(M_1)$.

% The objects of category $\MD$ are the equivalence classes
% $[M]\in \cM/\Z$ where $[M] = \{ M+n : n\in \Z \}$. The arrows are
% equivalence classes $[M,K]\in (\cM\times \hcZ)/\Z$ where
% $[M,K] = \{ (M+n,K+n) : n\in \Z \}$.  The source of the arrow $[M,K]$
% is $[M]$ and the target is $[f_K(M)]$.  This is well-defined,
% because
% \[ f_{(K+n)}(M+n) = f_K(M)+n,\quad M\in \cM,\; K\in \hcZ,\; n\in \Z.\]
% Composition of morphisms is given by the union of multi-sets:
% \[ [M_2,K_2] \circ [M_1,K_1] = [M_1, K_1 \cup K_2],\quad M_1\in
%   \cM,\; K_1,K_2\in \hcZ,\; M_2 = f_{K_1}(M_1). \]
% This definition makes sense because
% \[ f_{K_2}\circ f_{K_1} = f_{K_1 \cup K_2},\quad K_1,K_2\in \hcZ,\]
% and because
% \[ (K_1+n) \cup (K_2+n) = (K_1\cup K_2) +n,\quad K_1,K_2\in \hcZ,\;
%   n\in \Z.\]

% For a multi-set $K\in \hcZ$ and $n\in \Z$ define $K+n\in \hcZ$ by
% \[ (K+n)(i) = K(i-n),\quad i\in \Z.\]
% In effect, $K+n$ translates all elements of $K$  by $n$, so that if
% $K$ is expressed as in \eqref{eq:Kki}, then
% \[ K+n = \{ (k_1+n)^{p_1},\ldots, (k_q+n)^{p_q}\}.\]

% \begin{proposition}
%   \label{prop:homset}
%   Let $M_1,M_2\in \cM$ be Maya diagrams.  The corresponding $\Hom$ set
%  of the MD category admits the following
%   description:
%   \[
%     \Hom([M_1],[M_2])= \bigcup_{n\in \Z\atop K\in \hcZ_e} [M_1,
%     ((M_2+n) \ominus M_1)\cup K].\]
% \end{proposition}

For differential operators $A,T_1,T_2$, we say that $A$ intertwines
$T_1,T_2$ if
\[ A T_1 = T_2 A. \] The objects of $\REXT$ are the rational
extensions $T_M,\; M\in \cM$, and the arrows are monic differential
operators that intertwine two rational extensions.  Observe that if
$A$ intertwines $T_1, T_2$ then so does $A\circ p(T_1)$, where $p(x)$
is an arbitrary polynomial.  Given $T_1, T_2$, we say that $A$ is a
primitive intertwiner if it does not include a nontrivial right factor
$p(T_1)$.

For a Maya diagram
$M\in \cM$ and a set $K\in \cZ_p$, we define the operator
\[ A_{M,K}[y] = \frac{\Wr[\psi_{M,k_1},\ldots,
    \psi_{M,k_p},y]}{\Wr[\psi_{M,k_1},\ldots, \psi_{M,k_p}]} .\] By
construction, $A_{M,K}$ is a monic differential operator of order
$p$. These intertwining operators have their origin in SUSYQM
(supersymmetric quantum mechanics), and
obey the intertwining relation
\[ A_{M_1,K} T_{M_1} = T_{M_2} A_{M_1,K},\quad M_2 = f_K(M_1),\quad
  M_1,M_2\in \cM,\; K\in \cZ.\]

It is possible to show that $A_{M,K}$ is a primitive intertwiner
between $T_M$ and $T_{f_K(M)}$.  Moreover, it is possible to show
\cite[Proof of Theorem 3.10]{GFGUM} that that every arrow in $\REXT$
has the form $A_{M,K}\circ p(T_M)$, where $A_{M,K}$ is primitive
(i.e., $K$ is a set), and $p(x)$ is a polynomial. We also note that
these intertwiners are translation invariant:
\begin{equation}
  \label{eq:AMK+n}
   A_{M+n,K+n} = A_{M,K},\quad n\in \Z.
\end{equation}

In order to describe the composition of intertwiners, we need to
extend the above definition to include multi-sets.  For $K\in \hcZ$,
let $K_0\in \cZ$ and $K_1\in \hcZ$ be as per \eqref{eq:KK0K1}.  For
$M\in \cM$, we now define
\begin{equation}
  \label{eq:AMKK0}
  A_{M,K} = A_{M,K_0}\circ \prod_{k\in K_1} (2k+1-T_M).
\end{equation}
In other words, if $K\in \hcZ$ contains elements of higher
multiplicity, then $A_{M,K}$ is no longer primitive.  The arrows of
$\REXT$ are the operators $A_{M,K},\; M\in \cM,\; K\in \hcZ$.
Composition of arrows is just the usual composition of differential
operators.
\begin{theorem}
  \label{thm:functor}
  The correspondence $M\mapsto T_M,\, M\in \cM$ and
  $(M,K)\mapsto A_{M,K},\, K\in \hcZ$ is a covariant
  functor $\MD\to \REXT$.
\end{theorem}
\begin{proof}
  It suffices to observe that for $M_1\in \cM,\; K_1,K_2\in \hcZ$ we
  have
  \[ A_{M_2,K_2}\circ A_{M_1,K_1} = A_{M_1,K_1\cup K_2},\quad M_2 =
    f_{K_1}(M_1).\]
\end{proof}

\section{Ladder operators}
We define a ladder operator to be an intertwiner $A$ such that
\[ A T_M = (T_M+\lambda) A \] for some $M\in \cM$ and constant
$\lambda$.  Since $T_{M+n} = T_M+2n$, Theorem \ref{thm:functor}
implies that for every rational extension $T_M,\;M\in \cM$, and
$n\in \Z$, there exists a ladder operator $A_{M,K}$, where
$K = (M +n)\ominus M$.  By \eqref{eq:AMK+n} no generality is lost if
we index such ladder operators in terms of unlabelled Maya diagrams
$[M]\in \cM/Z$.

A recent result provides a characterization of translational
multi-flips \cite{GGM4} in terms of cyclic Maya diagrams. This
characterization makes it possible to establish the order of a ladder
operator \cite{GGM3}.
\begin{theorem}
  \label{thm:M+nM}
  Let  $M\in \cM$ and $n =1,2,\ldots$.
  Then,   
  \begin{equation}
    \label{eq:pngi}
    |(M+n) \ominus M| = n + 2\sum_{i=0}^{n-1} g_i ,
  \end{equation}
  where $g_i$ is the genus of the Maya diagram
  \[ M_i = \{ m\in \Z \colon mn + i \in M \},\quad i=0,1,\ldots, n-1
    ,\]
\end{theorem}

\begin{proof}
  Let $B_i\in \cZ_{2g_i+1}$ be the block coordinates of $M_i$, and set
  \[ B = \bigcup_{i=0}^{n-1} (nB_i+i) = \bigcup_{i=0}^{n-1} \{ nb + i
    \colon b \in B_i\}. \] Since $B_i$ is the unique set such that
  $f_{B_i}(M_i) = M_i+1$, it follows that $B$ is the unique set such
  that $f_B(M) = M+n$.  Therefore $B=(M+n)\ominus M$.
\end{proof}

Fix a Maya diagram $M\in \cM$. An immediate consequence of Theorem
\ref{thm:M+nM} is the existence of a primitive ladder operator that
intertwines $T_M$ and $T_M+2n$ for every $n\in \Z$.  The ladder
operator in question is $L_{n}:=A_{M,K}$, where $K=(M+n)\ominus M$.
The order of $L_{n}$ is given by \eqref{eq:pngi}. If $n>0$, then both
$L_{n}$ and $L_{1}^n$ intertwine $T_M$ and $T_M+2n$; it follows that
there must be a syzygy of the form
\[ L_{1}^n = L_{n} \circ p(T_M),\] where the roots of the polynomial $p$
are determined by \eqref{eq:AMKK0}.

The action of ladder operators on states is that of a lowering
or raising operator according to
\[ L_n[\psi_{M,k}] = C_{M,n,k} \psi_{M,k-n},\quad k\notin M, \] where
$C_{M,n,k}$ is zero if $\psi_{M,k-n}$ is not a bound state, i.e., if
$k-n\in M$. Otherwise, $C_{M,n,k}$ is a rational number whose explicit
form can be derived on the basis of \eqref{eq:hHdef}. As a particular
example, suppose that the index set of $M$ consists of positive
integers $0<k_1<\cdots < k_p$, that $n>0$, and that $k\notin M$. In
this case,
\[ C_{M,n,k} =
  \begin{cases}
    \prod_{i\in M\setminus(M+n)} (2i-2j) \times(k-n+1)_n\, 2^n &
    \text{ if } k-n \notin M\\
        0 &\text{ otherwise.}  
  \end{cases}
\]

\section{Examples} The articles \cite{MQ1,MQ2} considered a particular
class of ladder operators corresponding to Maya diagrams obtained by a
single state-adding transformation.  Fix some $n=1,2,\ldots$, and let
$\tM_n$ be the Maya diagram with index set $\{ -n\}$, i.e., let
$\tM_n = \Z_{-} \setminus \{ -n \}$. We set
\[ \hM_n = \tM_n+n = \Z_{-} \cup \{ 1,\ldots, n-1 \},\] and observe
that $\hM_n$ has index set $\{ 1,\ldots, n-1\}$. Hence, 
\[ L_n := A_{\tM_n,\{ -n,1,\ldots, n-1 \}},\] is an $n$th order
ladder operator that intertwines $T_{\tM_n}$ and $T_{\hM_n}$.  Ordering
the flips in ascending order, we obtain the following factorization
into first-order intertwiners:
\[ L_n = A_{\hM_{n-1},\{n-1\}} \cdots
  A_{\hM_2,\{2\}} A_{\hM_1,\{1\}} A_{\tM_n,\{-n\}} ;\] each flip
corresponds to a state-deleting transformation.

Let us also observe that $\tM_n$ is a genus 1 Maya diagram.  It follows
that
\[ L_{1} := A_{\tM_n,\{-n,-n+1,0\}}.\]
is a third-order ladder operator that intertwines
$\tM_n$ and $\tM_n+1$.  

The composition $L_{1}^n$ is represented by the multi-set
\[ \bigcup_{j=0}^{n-1}\{ -n+j,-n+j+1,j\} = \{ -n,1,\ldots, n-1\} \cup
  \{ (-n+1)^2,\ldots, (-1)^2, (0)^2 \},\] where the superscripts
indicate repetition (and not a square). The syzygy between $L_{n}$
and $L_{1}$ is therefore
\[ L_{1}^n = L_{n} \prod_{j=-n+1}^0 (2j+1-T_{\tM_n}) .\]

\end{document}